\documentclass[conference]{IEEEtran}
\IEEEoverridecommandlockouts

\usepackage{amsmath, amsthm, amssymb, amsfonts}
\usepackage[utf8]{inputenc} 
\usepackage[T1]{fontenc}    
\usepackage{hyperref}       
\usepackage{url}            
\usepackage{nicefrac}       
\usepackage{microtype}      
\usepackage{lipsum}
\usepackage{graphicx}       
\usepackage{multirow}
\usepackage{cite} 
\usepackage[table]{xcolor}
\usepackage{svg} 
\let\oldoplus\oplus
\renewcommand{\oplus}{\mathop{\oldoplus}\displaylimits}
\newcommand{\bracket}[1]{\left(#1\right)}
\newcommand{\signx}{s_X}
\newcommand{\signz}{s_Z}

\def\cS{\mathcal{S}}
\def\css{\textnormal{CSS}}
\usepackage{comment}
\usepackage{physics}
\theoremstyle{definition}
\usepackage{mathtools}
\usepackage{amsthm}
\theoremstyle{plain}
\newtheorem{theorem}{Theorem}
\newtheorem{lemma}[theorem]{Lemma}
\newtheorem{corollary}[theorem]{Corollary}

\newtheorem{example}[theorem]{Example}
\newtheorem{proposition}[theorem]{Proposition}
\theoremstyle{remark}
\newtheorem{remark}[theorem]{Remark}

\usepackage{cite}
\usepackage{amsmath,amssymb,amsfonts}
\usepackage{algorithmic}
\usepackage{graphicx}
\usepackage{textcomp}
\usepackage{xcolor}
\def\BibTeX{{\rm B\kern-.05em{\sc i\kern-.025em b}\kern-.08em
    T\kern-.1667em\lower.7ex\hbox{E}\kern-.125emX}}

\definecolor{newred}{HTML}{ff382e}
\definecolor{brickred}{HTML}{aa4a44}

\begin{document}

\title{Asymptotically good CSS codes that realize the logical transversal Clifford group fault-tolerantly}

\author{\IEEEauthorblockN{K. Sai Mineesh Reddy and Navin Kashyap}
\IEEEauthorblockA{
{Dept. of ECE}, IISc, Bengaluru \\
\{mineeshk, nkashyap\}@iisc.ac.in}
}

\maketitle

\begin{abstract}
This paper introduces a framework for constructing Calderbank-Shor-Steane (CSS) codes that support fault-tolerant logical transversal $Z$-rotations. Using this framework, we obtain asymptotically good CSS codes that fault-tolerantly realize the logical transversal Clifford group (i.e., transversal single-qubit Clifford gates are realized within a single code block, while transversal two-qubit Clifford gates are realized across two identical code blocks). Furthermore, investigating CSS-T codes, we: (a) demonstrate asymptotically good CSS-T codes wherein the transversal $T$ realizes the logical transversal $S^{\dagger}$; (b) show that the condition $C_2 \ast C_1 \subseteq C_1^{\perp}$ is necessary but not sufficient for CSS-T codes; and (c) revise the characterizations of CSS-T codes wherein the transversal $T$ implements the logical identity and the logical transversal $T$, respectively.
\end{abstract}


\section{Introduction}
Quantum error correcting codes are essential for fault-tolerant quantum computing~\cite{G1997}. An ${[[n,k,d]]}_2$ quantum code encodes $k$ logical qubits into $n$ physical qubits, correcting errors of weight up to $\lfloor \frac{d-1}{2} \rfloor$. A logical operator is a unitary acting on the physical qubits that preserves the code space, thereby realizing a logical gate on the encoded qubits. As physical gates required to implement logical operators are inherently noisy, their implementation must be fault-tolerant. A canonical approach is to employ transversal gates---tensor products of single-qubit unitaries---as they strictly limit error propagation. 

The Clifford group, together with a non-Clifford gate (such as $T$), constitutes a universal gate set for quantum computation (see Section~\ref{preliminaries_and_notation}). While identifying a code that realizes a universal set of logical gates via transversal physical gates is highly desirable, the Eastin-Knill Theorem~\cite{EK2009} precludes this.  In this work, we establish the existence of asymptotically good CSS codes that support the transversal implementation of the logical transversal Clifford group\footnote{By ``logical transversal Clifford group'', we mean that the transversal single-qubit Clifford gates are realized within the same code block and the transversal two-qubit Clifford gates are realized across two identical code blocks.\label{lt_clifford}}; however, by the Eastin-Knill theorem, these codes cannot realize the logical $T$ via transversal gates. This yields no quantum advantage, as the Gottesman-Knill theorem~\cite{NC2010} implies that circuits composed solely of Clifford unitaries can be simulated efficiently on classical computers.

Magic state distillation circumvents the Eastin-Knill theorem by utilizing magic states to realize non-Clifford gates~\cite{CTV2017}; these protocols rely on codes that fault-tolerantly support logical non-Clifford gates. While recent works constructed asymptotically good CSS codes that support the transversal implementation of the logical transversal controlled-controlled-$Z$ gate ($CCZ$, a non-Clifford gate) across three code blocks~\cite{WHY2025,GG2025,N2025}, the  logical $T$ gate remains crucial for several quantum algorithms.  

Consequently, identifying codes that fault-tolerantly realize the logical $T$ is of significant interest. Rengaswamy et al.~\cite{RCNP2020} characterized stabilizer codes in which the transversal $T$ is a logical operator. Moreover, they demonstrated that any ${[[n,k,d]]}_2$ non-degenerate stabilizer code admitting the transversal $T$ as a logical operator implies the existence of an ${[[n, \geq k,\geq d]]}_2$ CSS code with the same property. Therefore, the analysis is restricted to such CSS codes, termed CSS-T codes. A key open problem posed in~\cite{RCNPisit2020} asks whether there exist asymptotically good CSS-T codes wherein the transversal $T$ realizes some logical non-Clifford gate. While Berardini et al.~\cite{BDMJL2025} made initial progress by constructing asymptotically good CSS-T codes wherein the transversal $T$ acts as logical identity, we advance this inquiry by demonstrating asymptotically good CSS-T codes in which the transversal $T$ implements the logical transversal $S^{\dagger}$, a non-trivial Clifford gate. Unlike~\cite{WHY2025,GG2025,N2025}, our asymptotically good codes realize only Clifford gates; to realize non-Clifford gates, we pose an open problem in Section~\ref{conclusion}.

The paper is organized as follows. Section~\ref{preliminaries_and_notation} establishes the necessary preliminaries and notation. Section~\ref{asymptoticallygoodsection} develops a framework for constructing CSS codes that support fault-tolerant logical transversal $Z$-rotations, and leverages it to obtain the asymptotically good codes discussed previously. Section~\ref{css_t_codes} then resolves an open problem from~\cite{JA2025} and revises the characterizations of CSS-T codes wherein the transversal $T$ realizes the logical identity and the logical transversal $T$, respectively, addressing gaps in~\cite{RCNP2020}. Finally, Section~\ref{conclusion} summarizes the paper and outlines open problems.

\section{Notation and Preliminaries}\label{preliminaries_and_notation}
We adapt the notation and CSS construction from~\cite{RCNP2020,HLC2022}. The single-qubit Pauli operators are defined as:
\begin{align*}
    I_2 := \begin{bmatrix}
        1 & 0 \\ 0 & 1
    \end{bmatrix},\!\!\! \quad X := \begin{bmatrix}
        0 & 1 \\ 1 & 0
    \end{bmatrix},\!\!\! \quad Z := \begin{bmatrix}
        1 & 0 \\ 0 & -1
    \end{bmatrix},\!\!\! \quad Y := \iota XZ
\end{align*}
where $\iota = \sqrt{-1}$. The operators $X, Y,$ and $Z$ are Hermitian, unitary, and involutory (i.e., $X^2=Y^2=Z^2 = I_2$).

Let $\mathbb{F}_2$ denote the binary field. For any integer $n \geq 1$ and vectors $a = (a_1, \ldots, a_n), b = (b_1, \ldots, b_n) \in \mathbb{F}_2^{n}$, the $n$-qubit Pauli operators are defined as:
\begin{equation*}
    E(a,b) \; := \; \bracket{\iota^{a_1 b_1} X^{a_1} Z^{b_1}} \otimes \ldots \otimes \bracket{\iota^{a_n b_n} X^{a_n} Z^{b_n}}
\end{equation*}
where $\otimes$ denotes the Kronecker product. The $n$-qubit Pauli group is defined as the set:
\begin{equation*}
    \mathcal{P}_n  \; := \; \big\{ \iota^{\kappa} E(a, b) : a, b \in \mathbb{F}_2^n, \kappa \in \{0,1,2,3\} \big\}
\end{equation*}
The operators $E(a,b)$ are likewise Hermitian, unitary, and involutory. 

Let $N=2^n$, and let $\mathbb{U}_N$ denote the set of all $N \times N$ unitaries. The Clifford hierarchy is defined recursively: the first level is $ \mathcal{C}^{(1)} := \mathcal{P}_n$. For $l \geq 2$, the $l$-th level is defined as:
\begin{equation*}
    \mathcal{C}^{(l)} \; := \; \left\{   U \in \mathbb{U}_N : U P U^{\dagger} \in \mathcal{C}^{(l-1)} \; \forall \; P \in \mathcal{P}_n  \right\}
\end{equation*}
The second level $\mathcal{C}^{(2)}$, constitutes the Clifford group. It is generated by the Hadamard $(H)$, Phase $(S)$, and Controlled-NOT $(CNOT)$ gates, given by:
\begin{align*}
   H  := {\frac{1}{\sqrt{2}}} \begin{bmatrix}
      1 & 1 \\ 1 & -1  
    \end{bmatrix}, \! \quad S := \begin{bmatrix}
        1 & 0 \\ 0 & \iota
    \end{bmatrix}, \! \quad CNOT := \begin{bmatrix}
        I_2 & 0 \\ 0 & X
    \end{bmatrix}
\end{align*}
The controlled-$Z$ $(CZ)$ gate is also a Clifford gate given by\footnote{The set $\{H,S, CZ\}$ also generates the Clifford group.}:
\begin{equation*}
    CZ \; := \; \bracket{I_2 \otimes H}  \cdot CNOT \cdot \bracket{I_2 \otimes H} =  \begin{bmatrix}
        I_2 & 0 \\ 0 & Z
    \end{bmatrix}
\end{equation*}

As is well-known, the Clifford group forms a universal gate set for quantum computation when supplemented with any unitary from level $l \geq 3$ of the Clifford hierarchy. To this end, we study the family of $Z$-rotations, defined for integers $m \geq 0$: 
\begin{equation*}
    R_Z \bracket{\frac{\pi}{2^m}} \; := \; \begin{bmatrix}
        1 & 0 \\ 0 & e^{\iota \pi / 2^m}
    \end{bmatrix} \in \mathcal{C}^{(m+1)}
\end{equation*}
For $m \geq 2$, these rotations are non-Clifford. Section~\ref{css_t_codes} specifically focuses on the $T$ gate, defined as $T = R_Z(\frac{\pi}{4}) \in \mathcal{C}^{(3)}$.

We now establish standard notation for binary linear codes. For a code $C$, we denote its dimension by $\dim(C)$ and its minimum distance by $d_{\min}(C)$. Furthermore, we classify $C$ as a $2^n$-divisible code if the Hamming weight of every codeword $x \in C$ satisfies $w_H(x) = 0 \pmod{2^n}$. 

The Schur product of any two vectors $x = (x_1, \ldots, x_n), y=(y_1, \ldots, y_n) \in \mathbb{F}_2^n$ is defined as:
\begin{equation*}
    x * y \; := \; (x_1 y_1, \ldots, x_n y_n)
\end{equation*}
The Schur product of two codes $C_1$ and $C_2$ is defined as:
\begin{equation*}
    C_1 * C_2 \; := \; \left< x * y : x \in C_1, y \in C_2 \right> 
\end{equation*}
where $\left< J \right>$ denotes the $\mathbb{F}_2$-linear span of vectors in $J$. 

Let $[n] := \{ 1, \ldots, n\}$. Consider the linear combination $ \oplus_{i=1}^n a_i x_i$, where $a_i \in \mathbb{F}_2$ and $x_i \in \mathbb{F}_2^n$. By the principle of inclusion-exclusion, it can be established via induction that:
\begin{align}
    w_H&\Big(\oplus_{i=1}^{n} a_i x_i \Big) \; =  \nonumber \\
    &\sum_{i=1}^{n} {(-2)}^{i-1} \; \Bigg( \sum_{ \substack{\{j_1,  \ldots, j_i\} \\ \subseteq [n]}} \Big( \prod_{m=1}^{i} a_{j_m} \Big) \; w_H(x_{j_1} \ast  \ldots \ast x_{j_i}) \Bigg) \label{eqn_1}
\end{align}
Note that we use $\oplus$ for binary addition and $+$ for integer addition. If $C$ is a $2^n$-divisible code then~\eqref{eqn_1} implies that for any $i \in [n]$ and codewords $x_1, \ldots, x_i \in C$:
\begin{equation}\label{eqn_12}
    w_H(x_1 \ast  \ldots \ast x_i) \; = \; 0 \! \pmod{2^{n-i+1}}
\end{equation}
Next, we outline the CSS construction and its encoding. Consider a pair of binary linear codes $C_2 \subseteq C_1$, where $C_1$ is an ${[n, k_1]}_2$ code, and $C_2$ an ${[n,k_2]}_2$ code. The $X$-and $Z$-stabilizers of the CSS code are generated by $C_2$ and $C_1^\perp$, respectively.

Stabilizer signs play a crucial role in determining whether a unitary acts as logical operator. Given a CSS pair, $C_2 \subseteq C_1$, we parameterize the signs of $X$- and $Z$-stabilizers using signature vectors $s_X \in \mathbb{F}_2^n / C_2^{\perp} $ and $ s_Z \in \mathbb{F}_2^n / C_1$, respectively. The stabilizer group 
$\cS$ is generated by the $X$-stabilizer generators ${(-1)}^{w_H(s_X \ast x)} E(x, 0)$, $x \in C_2$, and the $Z$-stabilizer generators ${(-1)}^{w_H(s_Z \ast z)} E(0, z)$, $z \in C_1^{\perp}$.
For a fixed pair $C_2 \subseteq C_1$, distinct signatures yield distinct CSS codes. The stabilizer group $\cS$ defines an ${[[n, k:= k_1 - k_2]]}_2$ code, denoted by $\css{\bracket{C_1,C_2, s_X, s_Z}}$ or $\mathcal{Q}_{\cS}$. Its minimum distance is given by:
\begin{align}\label{eqn_18}
    d_{\min}(\mathcal{Q}_{\cS}) \; := \;  
    \min \Big\{ \;  &d_X \; := \; \underset{x  \in C_1 \setminus C_2}{\min} w_H(x), \nonumber \\
     &d_Z \; := \; \underset{z \in C_2^{\perp} \setminus C_1^{\perp}}{\min} w_H(z) \; \Big\}
\end{align}
Note that $C_1 \setminus C_2$ denotes the set difference and $C_1 / C_2$ denotes the coset space\footnote{Throughout this paper, the choice of coset representatives is not unique; any such choice should be regarded as arbitrary but fixed.}. The code projector, $P_{\cS}$, is defined as: 
\begin{align*}
    P_{\cS} \; := \; \frac{1}{{|C_2|}} \; &\Bigg(\sum_{x \in C_2} {(-1)}^{w_H(\signx \ast  x)} E(x, 0) \Bigg) \nonumber \\
      \frac{1}{|C_1^{\perp}|} \; &\Bigg(\sum_{z \in C_1^{\perp}} {(-1)}^{w_H(\signz \ast z)} E(0,z) \Bigg) 
\end{align*}
The image of $P_{\cS}$ is precisely the code space, $\mathcal{Q}_{\cS}$. Generally, a code state in $\mathcal{Q}_{\cS}$ is obtained by projecting a vector $\ket{\psi} \in {\bracket{\mathbb{C}^2}}^{\otimes n}$. However, since $Z$-stabilizers are signed, $\ket{\psi}$ shall be chosen carefully in order to obtain a non-zero code state; a suitable choice is $\ket{s_Z}$. The logical basis state $\ket{0^k}$ is encoded as:
\begin{align}
    \ket{0^k}_L \; &:= \; \sqrt{|C_2|} \;  P_{\cS} \ket{s_Z} \nonumber \\
   &\; = \; \frac{1}{\sqrt{|C_2|}} \; \Bigg( \sum_{x \in C_2} {(-1)}^{w_H(s_X \ast x)} \ket{x \oplus s_Z} \Bigg)  \label{eqn_20} 
\end{align}

Let $\{ y_1, \ldots, y_k \}$ be a basis of the coset space $C_1 / C_2$. For each $i \in [k]$, the logical-$X$ operator $\overline{X}_i$, which acts as the Pauli $X$ gate on the $i$-th logical qubit, is given by: $\overline{X}_i \; := \; E(y_i,0)$. For any $a = (a_1, \ldots, a_k) \in \mathbb{F}_2^k$, let $y_a := \oplus_{i=1}^k a_i y_i \in C_1 / C_2$. The logical basis state $\ket{a}$ is encoded as:
\begin{align}
    \ket{a}_L \; &:= \; \overline{X}_1^{a_1}  \ldots  \; \overline{X}_k^{a_k}  \ket{0^k}_L \; = \; E(y_a, 0) \ket{0^k}_L  \nonumber \\
    &\; = \; \frac{1}{\sqrt{|C_2 |}} \; \Bigg( \sum_{x \in C_2} {(-1)}^{w_H(s_X \ast x)} \ket{y_a \oplus x \oplus s_Z} \Bigg) \label{eqn_24} 
\end{align}

\section{Asymptotically good CSS codes supporting fault-tolerant logical Clifford unitaries}\label{asymptoticallygoodsection}
This section details the framework utilizing classical divisible codes to construct CSS codes that realize logical transversal $Z$-rotations via physical transversal $Z$-rotations. We utilize this framework to demonstrate the existence of the desired asymptotically good codes. 

In this section, we restrict our attention to CSS codes with positively signed stabilizers (i.e., $s_X=s_Z=0$); consequently, $s_X$ and $s_Z$ will be omitted from the CSS notation.

\subsection{CSS Construction using classical divisible codes}\label{css_construction}
For an integer $m \geq 2$, let $C$ be an ${[n,k,d]}_2$ $2^m$-divisible code with dual distance $d^{\perp}$. Let $t = \left\lfloor \frac12 \min \{ k, d, d^{\perp}\} \right \rfloor$.  The construction consists of two steps. 
\begin{itemize}
    \item \textit{Step~$1$}: we employ the technique introduced by Krishna and Tillich in~\cite{KT2019}: puncture $t$ coordinates of $C$ to obtain an  ${[[n - t, t, \geq t ]]}_2$ CSS code, $\css{(C_1, C_2)}$.\footnote{Any choice of $t < \min\{k,d, d^{\perp}\}$ suffices for the construction. The specific choice $t = \left\lfloor \frac{\min \{ k, d, d^{\perp}\}}{2} \right \rfloor$ is for the sake of clarity in exposition.}

    \item \textit{Step~$2$}: we generalize the doubling ($2$-fold repetition) technique introduced by Betsumiya and Munemasa in~\cite{BM2012}: for $p \geq 0$, the $2^p$-fold repetition of $C_1$ and $C_2$ yields a ${[[ 2^p (n - t) , t, \geq t ]]}_2$ CSS code, $\css{{(C_1^{(p)}, {C}_2^{(p)})}}$.
\end{itemize}

Jain and Albert~\cite{JA2025} constructed weakly-triply even codes by puncturing a coordinate of self-dual doubly-even ($4$-divisible) code, followed by doubling. While we generalize this method, the motivation differs; weakly-triply even codes are used to construct CSS codes that encode one logical qubit and realize logical-$T$ via a transversal gate composed of $T$ and $T^{\dag}$ gates.

\textit{Step $1$}: Let $G_{C}$ be a generator matrix of $C$; without loss of generality, assume $G_C$ is in systematic form: 
\begin{equation}\label{eqn_49}
    G_C = \left[\begin{array}{c|c} I_k & A \end{array}\right] \ \ \text{where } A \in \mathbb{F}_2^{k \times (n-k)}
\end{equation}
The CSS construction is specified via the codes $C_1$ and $C_2$ described below, and taking $s_X = s_Z = 0$. Puncturing the first $t$ coordinates of $C$ yields the code ${C}_1$ whose generator matrix $G_{{C}_1}$ is given by:
\renewcommand{\arraystretch}{1.5}
\begin{equation*}
 G_{C_1} = \left[\begin{array}{c|c} 0_{t \times (k - t)} & A_1 \\  \hline  I_{k - t} & A_2 \end{array}\right] \ \ \text{where } A = \left[\begin{array}{c} A_1 \\ \hline  {A}_2 \end{array}\right] 
\end{equation*}
The generator matrix of the coset space $C_1/C_2$ ($G_{C_1/C_2}$) constitutes the first $t$ rows of $G_{C_1}$, while the generator matrix of the code $C_2$ ($G_{C_2}$) constitutes the remaining $k-t$ rows:
\begin{align}\label{eqn_51}
   G_{C_1} =  \begin{array}{c@{}l@{}l} 
                \left[
                \begin{array}{c|c}
                    0_{t \times \bracket{k-t}} & A_1 \\ \hline        I_{k-t} & {A}_2
                \end{array}
                \right]
                &
                \begin{array}{l}
                    \left.
                    \vphantom{\begin{array}{c}  1 \end{array}}
                    \right\} \text{$\; = G_{C_1 / C_2}$}
                    \\
                    \left.
                    \vphantom{\begin{array}{c}  2 \end{array}} 
                    \right\} \text{$\; =G_{C_2}$} 
                \end{array}
            \end{array} 
\end{align}
The choice of generator matrices ensures $C_2 \subseteq C_1$, thereby yielding the CSS code, $\css{(C_1, C_2)}$. The following lemma describes the parameters of $C_1$ and $C_2$.
\begin{lemma}\label{lemma_7} The following properties hold:
\begin{enumerate}
    \item \label{item:first}  $\dim(C_1) = \dim (C) = k$ and $\dim(C_2) = k - t$.
    \item \label{item:second}  $d_{\min}(C_1) \geq t$ and $d_{\min}(C_2^{\perp}) \geq t$.
    \item If $C$ is self-dual then $C_2 = C_1^{\perp}$.
\end{enumerate}
\end{lemma}
\begin{proof}
    The proof is given in Appendix~\ref{appendix_lemma_1}.
\end{proof}
\begin{remark}\label{remark_10}
    Since $C_2 \subseteq C_1$ and $C_1^{\perp} \subseteq C_2^{\perp}$, Lemma~\ref{lemma_7} implies that each non-trivial stabilizer has weight at least $t$. 
\end{remark}
The next lemma concludes the first step.
\begin{lemma}\label{lemma_9a}
    $\css{(C_1, C_2)}$ is an ${[[n - t, t, \geq t]]}_2$ code. 
\end{lemma}
\begin{proof}
As $C_1$ and $C_2$ have block length $n - t$, the CSS code is defined on $n-t$ physical qubits. By Lemma~\ref{lemma_7}, it encodes $\dim(C_1) - \dim(C_2) = t$ logical qubits. Moreover, we have:
\begin{align*}
    d_X \; &= \; \underset{x  \in C_1 \setminus C_2}{\min} w_H(x) \; \geq \; d_{\min}(C_1) \; \geq \; t  \\
    d_Z \; &= \;  \underset{z \in C_2^{\perp} \setminus C_1^{\perp}}{\min} w_H(z)  \geq \; d_{\min}(C_2^{\perp}) \geq \; t 
\end{align*}
so that
$d_{\min}(\mathcal{Q}_{\cS}) = \min\{d_X, d_Z\} \geq t$. 
\end{proof}

\textit{Step $2$}: For any integer $p \geq 0$, let $C_1^{(p)}$ and ${C}_2^{(p)}$ denote the $2^p$-fold repetition codes of $C_1$ and $C_2$, respectively:
\begin{align*}
    C_1^{(p)} &:= \big\{ (\underbrace{y, \ldots,  y}_{2^p \text{ times}}) : y \in C_1 \big\}  \\
    C_2^{(p)} &:= \big\{ (\underbrace{x, \ldots, x}_{2^p \text{ times}}) : x \in C_2 \big\} 
\end{align*}
The inclusion $C_2 \subseteq C_1$ implies $C_2^{(p)} \subseteq C_1^{(p)}$, yielding the CSS code, $\css{(C_1^{(p)}, C_2^{(p)})}$. By~\eqref{eqn_18}, its minimum distance is $d_{\min}(\mathcal{Q}_{\cS}^{(p)}) = \min\{ d_X^{(p)}, d_Z^{(p)} \}$, where:
\begin{align*}
     {d}_X^{(p)} \; &= \ \ \; \; \underset{{x}  \in {C}_1^{(p)} \setminus {C}_2^{(p)}}{\min} w_H( {x} ) \nonumber \\
    d_Z^{(p)} \; &=  \underset{{z} \in {({C}_2^{(p)})}^{\perp} \setminus {({C}_1^{(p)})}^{\perp}}{\min} w_H({z})  
\end{align*}
The next lemma describes the parameters of $C_1^{(p)}$ and $C_2^{(p)}$.
\begin{lemma}\label{lemma_11} The following properties hold:
    \begin{enumerate}
        \item \label{item:lemma_11_first} $\dim(C_1^{(p)}) = \dim(C_1)$ and $\dim(C_2^{(p)}) = \dim(C_2)$.
        \item \label{item:lemma_11_second} ${d}_X^{(p)} = 2^p d_X$ and ${d}_Z^{(p)} = d_Z$.
    \end{enumerate}
\end{lemma}
\begin{proof}
    The proof is given in Appendix~\ref{appendix_lemma_1}. 
\end{proof}
\begin{remark}\label{remark_13}
    For any $x \in C_2$, $({x, \ldots, x}) \in {C}_2^{(p)}$, and for any $z \in C_1^{\perp}$, $(z, 0, \ldots, 0) \in {({C}_1^{(p)})}^{\perp}$. Thus, by Remark~\ref{remark_10}, there exist stabilizers of weight $\geq t$. 
\end{remark}
The following lemma concludes the second step.
\begin{lemma}\label{lemma_13}
    $\css{{({C}_1^{(p)}, {C}_2^{(p)})}}$ is a ${[[2^p (n-t), t, \geq t]]}_2$ code.
\end{lemma}
\begin{proof}
    As $C_1^{(p)}$ and $C_2^{(p)}$ have block length $2^p(n - t)$, the CSS code is defined on $2^p(n-t)$ physical qubits. By Lemma~\ref{lemma_7} and Lemma~\ref{lemma_11}, it encodes $\dim({C}_1^{(p)}) - \dim({C}_2^{(p)}) = t$ logical qubits and has minimum distance at least $t$.
\end{proof}

\subsection{$\css{(C_1^{(p)}, C_2^{(p)})}$ realizes logical transversal $Z$-rotations via physical transversal $Z$-rotations}
Hu et al.~\cite{HLC2025} proposed a framework to elevate the level of logical diagonal gates within the Clifford hierarchy. Beginning  with a CSS code wherein a physical level-$l$ diagonal gate induces a logical level-$l$ diagonal gate, their construction proceeds in three steps: $(1)$ double the code so that a physical level-$(l+1)$ diagonal gate induces the original logical level-$l$ gate; $(2)$ removing specific $Z$-stabilizers to promote the logical gate to level $(l+1)$; and $(3)$ adding specific $X$-stabilizers to compensate for the distance loss incurred in step~$(2)$.

The following lemma generalizes step~$(1)$ of~\cite{HLC2025}. While the primary objective of~\cite{HLC2025} is to increase the logical level, at the expense of code parameters, our approach differs. We generalize step~$(1)$ and omit steps~$(2)$ and~$(3)$; this yields good code parameters, although the logical level does not increase.

\begin{lemma}\label{lemma_16} In the code $\css{(C_1^{(p)}, C_2^{(p)})}$, the physical transversal $R_Z(\frac{\pi}{2^l})$ realizes the logical transversal ${\bracket{R_Z\bracket{\frac{\pi}{2^{l-p}}}}}^{\dagger}$ for $p \leq  l \leq m+p-1$, and the logical identity for $l \leq p - 1$.
\end{lemma}
\begin{proof}
    The proof is given in Appendix~\ref{appendix_lemma_1}.
\end{proof}

\subsection{Existence of asymptotically good CSS codes that fault-tolerantly realize logical transversal Clifford unitaries}
The rate and relative distance of an ${[[n,k,d]]}_2$ quantum code (or an ${[n,k,d]}_2$ classical code) are defined as $\frac{k}{n}$ and $\frac{d}{n}$, respectively. 

A sequence of quantum codes, ${({[[n_i, k_i, d_i]]}_2)}_{i \in \mathbb{N}}$ (respectively, classical codes, ${({[n_i, k_i, d_i]}_2)}_{i \in \mathbb{N}}$), is called asymptotically good if $\lim_{i \to \infty} n_i = \infty$, and both the asymptotic rate $\liminf_{i\to\infty} \frac{k_i}{n_i}$ and the asymptotic relative distance $\liminf_{i\to\infty} \frac{d_i}{n_i}$ are strictly positive. 

\begin{lemma}\label{lemma_8}
For each integer $l \geq 1$, there  exists a family of asymptotically good CSS codes in which the physical transversal $R_Z \bracket{\frac{\pi}{2^l}}$ realizes the logical transversal $S^{\dagger}$. 
\end{lemma}
\begin{proof}
It is known \cite{SMT1972}, \cite{PW2007} that there exist asymptotically good self-dual doubly-even codes, ${({[n_i, k_i = \frac{n_i}{2}, d_i]_2})}_{i \in \mathbb{N}}$, with asymptotic rate $\frac{1}{2}$ and relative distance $0 < \delta < \frac{1}{2}$. 

Define $t_i=\frac{\min \{ k_i, d_i \}}{2}$. For any $p \geq 0$, applying Lemma~\ref{lemma_13} yields a family of asymptotically good CSS codes, ${( {[[2^p(n_i - t_i), t_i, \geq t_i]]_2})}_{i \in \mathbb{N}}$, with asymptotic rate and relative distance at least $  \frac{\delta/2}{2^p (1 - (\delta/2))} > 0$.

Furthermore, by Lemma~\ref{lemma_16} (with $m=2$ and $l=p+1$), the physical transversal $R_Z\bracket{\frac{\pi}{2^{p+1}}}$ realizes the logical transversal ${R_Z\bracket{\frac{\pi}{2}}}^{\dagger} = S^{\dag}$ in these CSS codes.
\end{proof}
\begin{remark}
    Remark~\ref{remark_13} implies the existence of stabilizers whose weights scale linearly with the number of physical qubits. Consequently, the constructed family of asymptotically good CSS codes does not form a quantum LDPC family.
\end{remark}
A CSS code is called self-dual if its $X$-and $Z$-stabilizers are generated by the same code (i.e., $C_2 = C_1^{\perp}$). 
\begin{theorem}
    There exists a family of asymptotically good self-dual CSS codes that realize the logical transversal Clifford group via transversal physical gates.
\end{theorem}
\begin{proof}
    Lemma~\ref{lemma_8} (with $l=1$ and $p=0$) establishes the existence of a family of asymptotically good CSS codes, ${(\css(C_{1,i}, C_{2,i}))}_{i \in \mathbb{N}}$,  wherein the physical transversal $S$ realizes the logical transversal $S^{\dagger}$. The logical transversal $S$ is obtained via the relation $S = S^\dagger Z$, utilizing the fact that the physical transversal $Z$ realizes the logical transversal $Z$ (Lemma~\ref{lemma_16} with $l=0$ and $p=0$).
    
    By Lemma~\ref{lemma_7}, every code in this family is self-dual (i.e., $C_{2,i} = C_{1,i}^{\perp}$). Consequently, by~\cite{TTF2025}, the physical transversal $H$ realizes the logical transversal $H$ in these CSS codes. 
    
    Further, Appendix~\ref{transversal_cz} demonstrates that the physical transversal $CZ$ across two such code blocks realizes the logical transversal $CZ$. 
    
    Since the set $\{ S, H,  CZ\}$ generates the Clifford group, these CSS codes realize the logical transversal Clifford group.
\end{proof} 

Recall that CSS codes wherein the physical transversal $T$ is a logical operator are called CSS-T codes.
\begin{theorem}\label{theorem_11}
     There exist asymptotically good CSS-T codes wherein the transversal $T$ realizes the logical transversal $S^{\dagger}$.
\end{theorem}
\begin{proof}
    The proof follows by Lemma~\ref{lemma_8} with $l=2$, $p=1$.
\end{proof}

\section{CSS-T codes}\label{css_t_codes}
In this section, stabilizer signs play a pivotal role; thus, we reinstate our CSS notation with $X$-and $Z$-signatures. The following characterization of CSS-T codes is obtained by specializing Theorem~12 of~\cite{HLC2022} to the $T$ gate.
\begin{theorem}[Hu et al.~\cite{HLC2022}]\label{theorem_3}
    A CSS code ${\css(C_1, C_2, s_X, s_Z)}$ is a CSS-T  code iff  for all $x \in C_2$, $y \in C_1$,
    \begin{equation}
        w_H(x) - 2 w_H\big(x \ast \bracket{y \oplus s_Z} \big) = 0 \!\!\pmod{8}\label{eqn_25}
    \end{equation}
\end{theorem}
\begin{proof}
An independent proof is given in Appendix~\ref{appendix_2}.
\end{proof}
\begin{remark}
    Theorem~\ref{theorem_3} implies that the CSS-T characterization is independent of the $X$-stabilizer signs; consequently, we omit $s_X$ from the notation henceforth.
\end{remark}

The following lemma provides an equivalent characterization of CSS-T codes.
\begin{lemma}\label{lemma_21}
    Consider $C_2 \subseteq C_1$. Let $\{x_1, \ldots, x_{k_2}\}$ be a basis of $C_2$, and let $\{ y_1, \ldots, y_k\}$ be a basis of $C_1 / C_2$. A CSS code ${\css(C_1, C_2, s_Z)}$ is a CSS-T code iff for any distinct $i,j,p \in [k_2] $, distinct $q,r \in [k]$, the following hold:
        \begin{align}
            w_H(x_i) - 2 w_H(x_i \ast s_Z) &= 0 \!\pmod{8}\label{eqn_28f} \\
            w_H(x_i \ast x_j) - 2 w_H(x_i \ast x_j \ast s_Z) &= 0 \!\pmod{4} \label{eqn_29f} \\
            w_H(x_i \ast x_j \ast x_p)  &= 0 \!\pmod{2}  \label{eqn_31f}\\
             w_H(x_i \ast y_q) - 2 w_H(x_i \ast y_q \ast s_Z) &= 0 \!\pmod{4} \label{eqn_30f} \\
            w_H(x_i \ast x_j \ast y_q) &= 0 \!\pmod{2}\label{eqn_33f} \\
             w_H(x_i \ast y_q \ast y_r) &= 0 \!\pmod{2} \label{eqn_32f}
        \end{align}
\end{lemma}
\begin{proof}
    The proof is given in Appendix~\ref{appendix_2}.
\end{proof}
The following corollary is a consequence of Lemma~\ref{lemma_21}. This result was first noted in~\cite{RCNP2020} and subsequently proven in~\cite{CLMRSS2024}.
\begin{corollary}[Moreno et al.~\cite{CLMRSS2024}]\label{corollary_22}
If ${\css(C_1, C_2, s_Z)}$ is a CSS-T code for some $s_Z \in \mathbb{F}_2^n$, then $C_2 \ast C_1 \subseteq C_1^{\perp}$.
\end{corollary}
\begin{proof}
    The proof is given in Appendix~\ref{appendix_2}.
\end{proof}
An open problem was posed in~\cite{JA2025} that asked whether any pair $C_2 \subseteq C_1$ satisfying $C_2 \ast C_1 \subseteq C_1^{\perp}$ can form a CSS-T code ${\css(C_1, C_2, s_Z)}$ for some $s_Z$. The following example resolves this problem in the negative.

\begin{example}\label{ex:counterexample}
Consider $C_2 = \langle x_1,x_2,x_3,x_4,x_5 \rangle$ and $C_1 = \langle x_1,x_2,x_3,x_4,x_5, y_1 \rangle$,
where
    \begin{align*}
        y_1 := 000110110010100000101010101010101010 \\ 
        x_1 := 111111111100000000111111000011000000 \\
        x_2 := 111111000011110000111100110000110000 \\
        x_3 := 111100110011001100110000111100001100 \\
        x_4 := 111010101010101010110011001100000011 \\
        x_5 := 100111111111111111000000000000000000
     \end{align*}
It can be verified that for distinct $i, j, k \in [5]$,
\begin{equation*}
    w_H(x_i), \; w_H(x_i \ast x_j), \;  w_H(x_i \ast y_1) =  0 \pmod{2}  
\end{equation*}
\begin{equation*}
     w_H(x_i \ast x_j \ast x_k), \; w_H(x_i \ast x_j \ast y_1) = 0 \pmod{2}
\end{equation*}
The two equations above imply that $C_2 \ast C_1 \subseteq C_1^{\perp}$. Furthermore, we have
\begin{align}
x_5 \ast y_1 &= \bracket{x_1 \ast x_3} \oplus \bracket{x_2 \ast x_4} \nonumber \\
w_H(x_1 &\ast x_2 \ast x_3 \ast x_4 ) = 5 \label{eqn_39}
\end{align}

We now argue, by contradiction, that there exists no $s_Z \in \mathbb{F}_2^{36}$ such that ${\css\bracket{C_1, C_2, s_Z}}$ forms a CSS-T code. 
Suppose that there exists some $s_Z$ such that ${\css{(C_1, C_2, s_Z)}}$ forms a CSS-T code. Then, by~\eqref{eqn_30f}, we must have (modulo $4$)
    \begin{align*}
        0 &= w_H(x_5 \ast y_1) - 2 w_H(x_5 \ast y_1 \ast s_Z) \\
         &= w_H(x_1 \ast x_3) + w_H(x_2 \ast x_4) - 2 w_H(x_1 \ast x_2 \ast x_3 \ast x_4) \\
        & \ \ \ \ \ \ \ \ \ -2w_H(x_1 \ast x_3 \ast s_Z) - 2w_H(x_2 \ast x_4 \ast s_Z) 
    \end{align*}
    Now, using~\eqref{eqn_29f}, we obtain
    \begin{equation*}
        -2 w_H(x_1 \ast x_2 \ast x_3 \ast x_4 ) = 0 \pmod{4}
    \end{equation*}
    This contradicts~\eqref{eqn_39}, thereby completing the proof.
\end{example}

Consequently, $C_2 \ast C_1 \subseteq C_1^{\perp}$ is a necessary but not sufficient condition for CSS-T codes.

\subsection{Characterization of CSS-T codes wherein the transversal $T$ realizes the logical identity}
Rengaswamy et al.~\cite{RCNP2020} provided a characterization of CSS-T codes wherein the transversal $T$ realizes the logical identity: 

\textit{Suppose ${\css(C_1, C_2)}$ is a CSS-T code. The transversal $T$ realizes the logical identity iff for all $y_1, y_2 \in C_1$, 
\begin{equation*}
    \iota^{w_H(y_1 \ast y_2)} E(0, y_1 \ast y_2) \in \cS
\end{equation*}}

Applying Lemma~\ref{lemma_16} with $m=p=l=2$ yields a CSS-T code, $\css{(C_1^{(2)}, C_2^{(2)}, s_X=s_Z=0)}$, wherein the transversal $T$ realizes the logical transversal $Z$. However, this code satisfies the above characterization of~\cite{RCNP2020} (due to the $4$-fold repetition and the fact that ${C}_1^{(2)} \ast {C}_1^{(2)} \subseteq {({C}_1^{(2)})}^{\perp}$). There is thus a gap in the above characterization stemming from the fact that it does not account for the $Z$-signature.

The following theorem revises their characterization. An equivalent form of this result appears as a special case of Lemma~4 in~\cite{HLC2022div}.
\begin{theorem}\label{theorem_19}
    Suppose ${\css(C_1, C_2, s_Z)}$ is a CSS-T code. The transversal $T$ realizes the logical identity iff for all $y \in C_1 / C_2$,
    \begin{equation}\label{eqn_55}
        w_H(y) - 2 w_H(y \ast s_Z) = 0 \!\!\!\pmod{8}
    \end{equation}
\end{theorem}
\begin{proof}
    The proof is given in Appendix~\ref{appendix_d}.
\end{proof}

\subsection{Characterization of CSS-T codes wherein the transversal $T$ realizes the logical transversal $T$}
Rengaswamy et al.~\cite{RCNP2020} provided a characterization of CSS-T codes wherein the transversal~$T$ realizes the logical transversal $T$; however, this characterization also does not account for the $Z$-signature, $s_Z$. The following theorem revises their characterization. An equivalent form of this result appears as a special case of Theorem~5 in~\cite{HLC2022div}.
\begin{theorem}\label{theorem_20}
    Suppose $\css{(C_1, C_2, s_Z)}$ is an ${[[n,k]]}_2$ CSS-T code. Let $\{y_1, \ldots, y_k\}$ be a basis of the coset space ${C_1 / C_2}$. The transversal $T$ realizes the logical transversal $T$ (without any Clifford correction) iff for any $a=(a_1,\ldots,a_k) \in \mathbb{F}_2^k$ and $y_a = \oplus_{i=1}^k a_i y_i \in C_1 / C_2$,
    \begin{equation}\label{eqn_17}
        w_H(y_a) - 2 w_H(y_a \ast s_Z) \; = \; w_H(a) \!\!\!\pmod{8}
    \end{equation}
\end{theorem}
\begin{proof}
    The proof is given in Appendix~\ref{appendix_d}.
\end{proof}

\subsection{Comparison of asymptotically good CSS-T codes in~\cite{BDMJL2025} and Theorem~\ref{theorem_11}}
Asymptotically good CSS-T codes in~\cite{BDMJL2025} are obtained by doubling arbitrary families of asymptotically good CSS codes. While~\cite{BDMJL2025} verifies that these codes satisfy the condition $C_2 \ast C_1 \subseteq C_1^{\perp}$, it does not explicitly establish the existence of $s_Z$ such that $\css{(C_1, C_2, s_Z)}$ forms a CSS-T code.

\begin{lemma}\label{lemma_19}
    Let $C_2 \subseteq C_1$ be a CSS pair. Let ${C}_1^{(1)}$ and ${C}_2^{(1)}$ denote the codes obtained by doubling $C_1$ and $C_2$, respectively. Then, $\css{({C}_1^{(1)}, {C}_2^{(1)},s_Z =(1, 0))}$ is a CSS-T code wherein the transversal $T$ realizes the logical identity. Here, $1$ and $0$ denote the all-ones and all-zeros vectors of the same length.
\end{lemma}
\begin{proof}
The proof is given in Appendix~\ref{appendix_d}.
\end{proof}
\begin{remark}
        Theorem~\ref{theorem_11} constructs CSS-T codes by doubling CSS codes derived from punctured doubly-even codes. Hence, both $\css{({C}_1^{(1)}, {C}_2^{(1)}, s_Z=0)}$ and $\css{({C}_1^{(1)}, {C}_2^{(1)}, s_Z=({1}, {0}))}$ are valid CSS-T codes. We highlight that the choice of $s_Z$ determines the logical action of the transversal $T$: it realizes the logical transversal $S^{\dagger}$ in the former code, and the logical identity in the latter. 
\end{remark}

\section{Conclusion}\label{conclusion}
In this work, we demonstrated the existence of asymptotically good families of $(1)$ CSS codes that fault-tolerantly realize the logical transversal Clifford group, and $(2)$ CSS-T codes wherein the transversal $T$ implements the logical transversal $S^{\dagger}$. It is shown in~\cite{BDMJL2025} that the logical action induced by the transversal $T$ in any doubled CSS-T code has order dividing $4$; the codes in our asymptotically good CSS-T family achieve the maximum order (as the order of $S^{\dagger}$ is $4$). 

Ideally, one seeks asymptotically good CSS-T codes wherein the transversal $T$ realizes the logical transversal $T$. To this end, we pose the following open problem: \textit{Does there exist a family of asymptotically good triply-even ($8$-divisible) codes whose duals are also asymptotically good?} An affirmative answer implies the existence of asymptotically good CSS-T codes wherein the transversal $T$ realizes the logical transversal $T^{\dagger}$, via Lemma~\ref{lemma_16} (logical transversal $T$ can be derived using logical transversal $S$, as $T^\dagger S = T$). Furthermore, as the CSS families established in this work are not LDPC, the corresponding existence questions remain open in the LDPC setting.

In the latter half of this work, we analyzed CSS-T codes, establishing that the condition $C_2 \ast C_1 \subseteq C_1^{\perp}$ is necessary but not sufficient. We also revised the characterizations provided in~\cite{RCNP2020} for CSS-T codes wherein the transversal $T$ realizes the logical identity and the logical transversal $T$, respectively.

\bibliographystyle{ieeetr}
\nocite{*}
\bibliography{references}

\onecolumn

\section{Appendix}

\subsection{Proof of Lemma~\ref{lemma_7}, Lemma~\ref{lemma_11}, and Lemma~\ref{lemma_16}}\label{appendix_lemma_1}
The proposition below summarizes the properties of the matrix $A$ required for subsequent proofs.
\begin{proposition}\label{proposition_6} 
Let $r_1, \ldots, r_k$ denote the rows of $A$. The following statements hold:
\begin{enumerate} 
    \item For integers $\{ i_1, \ldots, i_q \} \subseteq [k]$ such that $2 \leq q \leq \min\{m, k\}$, the rows $r_{{i_1}}, \ldots, r_{i_{q}}$ of $A$ satisfies:
        \begin{equation}\label{eqn_49a}
            w_H(r_{i_1} \ast \ldots \ast r_{i_q}) \; = \; 0 \!\pmod{2^{m-q+1}}.
        \end{equation}
    \item Each row $r_i$ of $A$ satisfies:
    \begin{equation}\label{eqn_50}
        w_H(r_i) \; = \; -1 \!\pmod{2^m}.
    \end{equation}
\end{enumerate}
\end{proposition}

\begin{proof} \text{ }
\begin{enumerate}
    \item Since $A$ is a submatrix of the systematic generator matrix $G_C$, each row $r_i$ of $A$ extends to a codeword $\widetilde{r}_i = (e_i, r_i) \in C$, where $e_i$ denotes the standard basis vector. Since $C$ is a $2^m$-divisible code (and $m \geq 2$), applying~\eqref{eqn_12} yields:
    \begin{align*}
        0 \!\pmod{2^{m-q+1}} \; &= \; w_H(\widetilde{r}_{i_1} \ast \ldots \ast \widetilde{r}_{i_q}) \\
        &=\;w_H(e_{i_1} \ast \ldots \ast e_{i_q}) + w_H(r_{i_1} \ast \ldots \ast r_{i_q}) \\
        &=\; w_H(r_{i_1} \ast \ldots \ast r_{i_q}) 
    \end{align*}
    where the last equality follows because the vectors $e_i$ have pairwise disjoint supports for distinct rows.
    \item Since $C$ is a $2^m$-divisible code, for any $i \in [k]$, we have:
    \begin{equation*}
       0 \!\pmod{2^m} \; = \; w_H(\widetilde{r}_i) \; = \; w_H(e_i) + w_H(r_i) \; = \; 1 + w_H(r_i)
    \end{equation*}

\end{enumerate}
\end{proof}

The next proposition summarizes the properties of the codes $C_1$ and $C_2$ required for subsequent proofs.
\begin{proposition}\label{proposition_22} The following statements hold:
    \begin{enumerate}
        \item The code $C_2$ is $2^m$-divisible.
    \item For all $x \in C_2$, $y \in C_1$:
    \begin{equation*}
        w_H(x \ast y) \; = \; 0 \!\pmod{2^{m-1}}.
    \end{equation*}
    \item Let $y_1, \ldots, y_{t}$ be the rows of $G_{C_1/ C_2}$. For any $a = (a_1, \ldots, a_t) \in \mathbb{F}_2^{t}$, the coset representative $y_a = \oplus_{i=1}^{t} a_i y_i$ satisfies:
    \begin{equation*}
        w_H(y_a) \; = \; -w_H(a) \!\pmod{2^{m}}.
    \end{equation*}
    \end{enumerate}
\end{proposition}
\begin{proof} \text{ }
    \begin{enumerate}
    \item For any $x \in C_2$, the vector $(0^t, x)$ lies in $C$ by construction. Since $C$ is a $2^m$-divisible code, $C_2$ inherits this property. 

    \item For any $x\in C_2$, $y \in C_1$, there exists a $y' \in \mathbb{F}_2^{t}$ such that $(0^t,x), (y',y) \in C$. Applying~\eqref{eqn_12}, we get:
        \begin{equation*}
            0 \!\pmod{2^{m-1}} \; = \; w_H \big( (0^t, x) \ast {(y',y)} \big) \; = \; w_H(x \ast y)
        \end{equation*}

    \item As in Proposition~\ref{proposition_6}, let $r_1, \ldots, r_t$ denote the first $t$ rows of $A$. By construction, the rows of the matrix $G_{C_1/ C_2}$ can be expressed as: $y_1 = (0^{k-t}, r_1), \ldots, y_t = (0^{k-t}, r_t)$. Therefore, 
        \begin{align}
            w_H(y_a) &= \; \sum_{i=1}^{t} {(-2)}^{i-1} \Bigg( \sum_{ \substack{\{j_1,  \ldots, j_i\} \\ \subseteq [t]}} \Big( \prod_{p=1}^{i} a_{j_p} \Big) w_H(y_{j_1} \ast \ldots \ast y_{j_i}) \Bigg) \label{eqn_50a} \\
           &= \sum_{i=1}^{t} {(-2)}^{i-1} \Bigg( \sum_{ \substack{\{j_1,  \ldots, j_i\} \\ \subseteq [t]}} \Big( \prod_{p=1}^{i} a_{j_p} \Big) w_H(r_{j_1} \ast \ldots \ast r_{j_i}) \Bigg) \\
            &= \; \sum_{j=1}^t a_{j} \; w_H(r_j)  \!\pmod{2^m}  \label{eqn_51a} \\
            &= \; - w_H(a) \!\pmod{2^m} \label{eqn_52a}
        \end{align}
        where Eqs.~\eqref{eqn_50a},~\eqref{eqn_51a} and~\eqref{eqn_52a} follow from~\eqref{eqn_1},~\eqref{eqn_49a} and~\eqref{eqn_50}, respectively.
    \end{enumerate}
\end{proof}
\newpage
\begin{proof}[Proof of Lemma~\ref{lemma_7}] \text{ }
    \begin{enumerate}
        \item By construction (Eq.~\eqref{eqn_51}), we have $\dim(C_2) = k - t$. Furthermore, the code $C_1$ is obtained by puncturing $t < d_{\min}(C)$ coordinates from the code $C$. Because the number of punctured coordinates is strictly less than the minimum distance of $C$, the puncturing operation is bijective. Consequently, the cardinality of the code is preserved, yielding $\dim(C_1) = \dim(C) = k$.

        \item For any $x \in C_1$, by the puncturing construction, there exists $x' \in \mathbb{F}_2^{t}$ such that $(x', x) \in C$. Since $w_H(x') \leq t$ and $d_{\min}(C) = d$, we obtain:
            \begin{equation*}
                d_{\min}(C_1) \; \geq \;  d - t \; \geq \; t
            \end{equation*}
            Regarding $C_2^{\perp}$, observe that the generator matrix of $C_2^{\perp}$ $(G_{C_2^{\perp}})$ forms a submatrix of the generator matrix of $C^{\perp}$ $(G_{C^\perp})$:
            \begin{align*}
                G_{C_2^{\perp}}  &= \; \left[\begin{array}{c|c}  A_2^T & I_{n - k} \end{array}\right]  \\  G_{C^{\perp}} &= \; \left[\begin{array}{c|c} {A}^T & I_{n-k} \end{array}\right] \; = \; \left[\begin{array}{c|c} A_1^T & G_{C_2^{\perp}} \end{array}\right]
            \end{align*}
            Thus, any $z \in C_2^{\perp}$ extends to $(z',z) \in C^{\perp}$ for some $z' \in \mathbb{F}_2^t$. Since $w_H(z') \leq t$ and $d_{\min}(C^\perp)= d^{\perp}$, we obtain:
            \begin{equation*}
                d_{\min} (C_2^{\perp}) \; \geq \; d^{\perp} - t \; \geq \; t
            \end{equation*}

        \item Proposition~\ref{proposition_22} (Item~2) implies that the vectors in $C_2$ and $C_1$ are mutually orthogonal; hence, $C_2 \subseteq C_1^{\perp}$. If $C$ is self-dual then $k = \frac{n}{2}$. Therefore, $\dim(C_2) = \frac{n}{2} - t$ and $\dim(C_1^{\perp}) = n - t - k = \frac{n}{2} - t$, which implies that $C_2 = C_1^{\perp}$.
    \end{enumerate}
\end{proof}


\begin{proof}[Proof of Lemma~\ref{lemma_11}] \text{ }
    \begin{enumerate}
        \item The $2^p$-fold repetition preserves code dimensions; hence, $\dim(C_i^{(p)}) = \dim(C_i)$ for $i \in \{1, 2\}$.
        \item $d_X^{(p)} = 2^p \; d_X$ follows directly from the repetition construction. It remains to show $d_Z^{(p)} = d_Z$: 
        \begin{enumerate}
        \item ${d}_Z^{(p)} \geq d_Z$:
        Let $(z_1, \ldots, z_{2^p})$ be a minimum weight vector in $ {({C}_2^{(p)})}^{\perp} \setminus {({C}_1^{(p)})}^{\perp}$ then,
        \begin{enumerate}
            \item for any $(x, \ldots, x) \in C_2^{(p)}$, we have: 
            \begin{equation*}
                0 \!\pmod{2} \; = \; w_H \big( (z_1, \ldots, z_{2^p}) \ast (x, \ldots, x) \big) \; = \; w_H \big((\oplus_{i=1}^{2^p} z_i) \ast x \big) \implies \oplus_{i=1}^{2^p} z_i \in C_2^\perp
            \end{equation*}
            \item there exists $(y, \ldots, y) \in C_1^{(p)}$ such that:
            \begin{equation*}
                1 \!\pmod{2} \; = \; w_H \big( (z_1, \ldots, z_{2^p}) \ast (y, \ldots, y) \big) \; = \; w_H \big( (\oplus_{i=1}^{2^p} z_i) \ast y \big) \implies \oplus_{i=1}^{2^p} z_i \notin C_1^\perp
            \end{equation*}
            \item Therefore, $\oplus_{i=1}^{2^p} z_i \in C_2^\perp \setminus C_1^\perp$, which implies that: 
            \begin{equation*}
               {d}_Z^{(p)} = \; w_H(z_1, \ldots, z_{2^p}) \; \geq \; w_H(\oplus_{i=1}^{2^p} z_i) \; \geq \; d_Z
            \end{equation*}
        \end{enumerate}
        (note that $z_i, x, y \in \mathbb{F}_2^{n-t}$) \\[2pt]
        \item $d_Z \geq {d}_Z^{(p)}$: Let $z$ be a minimum weight vector in $C_2^\perp \setminus C_1^\perp$ then,
        \begin{enumerate}
            \item for any $x \in C_2$, we have:
            \begin{equation*}
                w_H\big( (z, 0, \ldots, 0) \ast (x, \ldots, x) \big) \; = \; w_H( z \ast x) \; = \; 0 \!\pmod{2} \implies (z,0, \ldots, 0) \in {({C}_2^{(p)})}^\perp
            \end{equation*}           
            \item there exists $y \in C_1$ such that, $w_H(z \ast y) = 1\! \pmod{2}$. As a result:
            \begin{equation*}
                w_H \big( (z,0, \ldots, 0) \ast (y, \ldots, y) \big) \; = \; w_H(z \ast y) \;= \; 1\! \pmod{2} 
                \implies (z,0, \ldots, 0) \notin {({C}_1^{(p)})}^\perp
            \end{equation*}
            \item Therefore, $(z,0, \ldots, 0) \in {({C}_2^{(p)})}^\perp \setminus {({C}_1^{(p)})}^\perp $, which implies that:
            \begin{equation*}
                d_Z \; = \; w_H(z) \; = \;  w_H(z,0, \ldots, 0) \; \geq \; {d}_Z^{(p)}
            \end{equation*}
        \end{enumerate}
    \end{enumerate}
    \end{enumerate}
\end{proof}
\newpage
\begin{proof}[Proof of Lemma~\ref{lemma_16}] \text{ }\\
    Let $y_1, \ldots, y_t$ denote the rows of $G_{C_1/C_2}$. Then, the set $\big\{ \widetilde{y}_i = (y_i, \ldots, y_i) : i \in [t] \big\}$ forms a basis of the coset space $C_1^{(p)}/C_2^{(p)}$. For any $a = (a_1, \ldots, a_{{t}}) \in \mathbb{F}_2^{t}$, let $ \widetilde{y}_a = \oplus_{i=1}^{t} a_i \widetilde{y}_i \in C_1^{(p)} / C_2^{(p)}$. Applying the transversal $R_Z(\frac{\pi}{2^{l}})$ to the logical state $\ket{a}$, we obtain:
    \begin{align*}
        {R_Z \bracket{ \frac{\pi}{2^l}}}^{\otimes \bracket{ 2^p(n-{t})}} \; \ket{a}_L  \; &= \; \frac{1}{\sqrt{|{C}_2^{(p)} |}} \; \Bigg( \sum_{\widetilde{x} \in {C}_2^{(p)}} e^{\iota \frac{\pi}{2^l} w_H( \widetilde{x} \oplus \widetilde{y}_a)} \ket{\widetilde{x} \oplus \widetilde{y}_a} \Bigg) \\
        &= \; e^{\iota \frac{\pi}{2^l} w_H(\widetilde{y}_a)} \; \frac{1}{\sqrt{| {C}_2^{(p)} |}}  \; \Bigg( \sum_{\widetilde{x} \in {C}_2^{(p)}} e^{\iota \frac{\pi}{2^l} \bracket{ w_H(\widetilde{x}) -2w_H(\widetilde{x} \ast \widetilde{y}_a)}} \ket{\widetilde{x} \oplus \widetilde{y}_a}  \Bigg)
    \end{align*}
    By construction, any $\widetilde{x} \in C_2^{(p)}$ and $\widetilde{y}_a \in C_1^{(p)} / C_2^{(p)}$ decompose as $\widetilde{x} = (x, \ldots, x)$ and $\widetilde{y}_a = (y_a, \ldots, y_a)$, where $x \in C_2$ and $y_a = \oplus_{i=1}^{t} a_i y_i \in C_1 / C_2$. Proposition~\ref{proposition_22} implies that:
    \begin{align*}
        &w_H(x) = 0 \!\pmod{2^m}  \\
        &2 w_H(x \ast y_a) = 0 \!\pmod{2^m}  \\
        &w_H(y_a) =  -w_H(a) \!\pmod{2^m} 
    \end{align*}
    From the condition $l+1 \leq m + p $ and the equations above, we deduce:
    \begin{align*}
        &w_H(\widetilde{x}) = 2^p w_H(x) = 0 \!\pmod{2^{m+p}} = 0 \!\pmod{2^{l+1}}  \\
       &2w_H(\widetilde{x} \ast \widetilde{y}_a) = 2^{p+1} w_H(x \ast y_a) = 0 \!\pmod{2^{m+p}} = 0 \!\pmod{2^{l+1}} \\
       &w_H(\widetilde{y}_a) = 2^p w_H(y_a) = -2^p w_H(a)  \!\pmod{2^{m+p}} = -2^p w_H(a) \!\pmod{2^{l+1}} 
    \end{align*}
    This implies that the transversal $R_Z \bracket{ \frac{\pi}{2^l}}$ is a logical operator with the following logical action:
    \begin{equation*}
        {R_Z \bracket{\frac{\pi}{2^l} }}^{\otimes \bracket{ 2^p(n-{t})}} \; \ket{a}_L \; = \; e^{\iota \frac{\pi}{2^l} (-2^p w_H(a))}  \ket{a}_L  
    \end{equation*}
    If $l \leq p - 1$, this implements the logical identity. Otherwise,
    \begin{equation*}
        {R_Z \bracket{\frac{\pi}{2^l}} }^{\otimes \bracket{2^p(n-{t})}} \ket{a}_L \; = \; e^{-\iota \frac{\pi}{2^{l-p}}w_H(a)} \ket{a}_L \; =  \;\overline{{\bracket{{R_{Z} \bracket{ \frac{\pi}{2^{l-p}}}}^\dagger}}^{\otimes t}} \; \ket{a}_L
    \end{equation*}
    where $\overline{{\bracket{{R_{Z} \bracket{ \frac{\pi}{2^{l-p}}}}^\dagger}}^{\otimes t}}$ denotes the logical transversal ${R_{Z} \bracket{ \frac{\pi}{2^{l-p}}}}^\dagger$. 
\end{proof}

\newpage
\subsection{Physical Transversal $CZ$ realizes logical transversal $CZ$ across two code blocks, $(\css{(C_1, C_2)}, \css{(C_1, C_2)})$}\label{transversal_cz}
Let $y_1, \ldots, y_t$ denote the rows of $G_{C_1 / C_2}$. For any $a=(a_1, \ldots, a_t), b=(b_1, \ldots, b_t) \in \mathbb{F}_2^t$, let $y_a = \oplus_{i=1}^t a_i y_i, y_b = \oplus_{j=1}^t b_j y_j \in C_1/ C_2$. Applying the physical transversal $CZ$ across two code blocks yields:
\begin{equation}\label{eqn_171}
    {CZ}^{\otimes (n-t)} \ket{a}_L  \ket{b}_L \; = \; \frac{1}{|C_2|} \; \sum_{x_1 \in C_2}  \sum_{x_2 \in C_2}  {(-1)}^{w_H((x_1 \oplus y_a) \ast (x_2 \oplus y_b))}  \ket{x_1 \oplus y_a}  \ket{x_2 \oplus y_b}
\end{equation}
Expanding the phase exponent, we obtain:
\begin{equation*}
    w_H\big((x_1 \oplus y_a) \ast (x_2 \oplus y_b)\big) = w_H(x_1 \ast x_2) + w_H(x_1 \ast y_b) + w_H(y_a \ast x_2) + w_H(y_a \ast y_b) \pmod{2}
\end{equation*}
Proposition~\ref{proposition_22} (Item~2) establishes that vectors in $C_2$ and $C_1$ are mutually orthogonal. Consequently:
\begin{align*}
    w_H \big((x_1 \oplus y_a) \ast (x_2 \oplus y_b) \big) \; &= \; w_H(y_a \ast y_b) \pmod{2} \\
    &= \; w_H \Big( (\oplus_{i=1}^t a_i y_i) \ast (\oplus_{j=1}^t b_j y_j) \Big) \pmod{2} \\
    &= \; w_H \Big( \oplus_{i=1}^{t} \oplus_{j=1}^t a_i b_j y_i \ast y_j \Big) \pmod{2} \\
    &= \; \sum_{i=1}^{t} \sum_{j=1}^t a_i b_j w_H(y_i \ast y_j) \pmod{2} 
\end{align*}

Let $r_1, \ldots, r_k$ denote the rows of $A$. Proposition~\ref{proposition_6} establishes that $$w_H(y_i \ast y_j) = w_H((0^{k-t}, r_i) \ast (0^{k-t}, r_j)) = w_H(r_i \ast r_j) = \delta_{ij} \pmod{2}$$
Therefore, the phase exponent simplifies to:
\begin{equation*}
    w_H\big((x_1 \oplus y_a) \ast (x_2 \oplus y_b)\big) \; = \; \sum_{i=1}^{t} a_i b_i \!\pmod{2} 
\end{equation*}
Substituting this back into~\eqref{eqn_171} yields:
\begin{align*}
    {CZ}^{\otimes (n-t)}  \ket{a}_L \ket{b}_L \; &= \;  {(-1)}^{\sum_{i=1}^t a_i b_i}  \ket{a}_L  \ket{b}_L \\
    &= \;\overline{{CZ}^{\otimes t}}  \ket{a}_L \ket{b}_L
\end{align*}
where $\overline{{CZ}^{\otimes t}}$ denotes the logical transversal $CZ$. 

\newpage
\subsection{Proof of Theorem~\ref{theorem_3}, Lemma~\ref{lemma_21}, and Corollary~\ref{corollary_22}}\label{appendix_2}
\begin{proof}[Proof of Theorem~\ref{theorem_3}] \text{ }\\
The proof proceeds by analyzing the action of $T^{\otimes n}$ on the logical basis states. Let $\{y_1, \ldots, y_k\}$ be a basis of $C_1 / C_2$\footnote{Note that the choice of basis is immaterial to the proof}. 

\textit{Sufficiency}: For any $a = (a_1, \ldots, a_k) \in \mathbb{F}_2^k$, let $y_a = \oplus_{i=1}^k a_i y_i \in C_1 / C_2$. Applying $T^{\otimes n}$ to the logical state $\ket{a}$ in~\eqref{eqn_24} yields:
\begin{equation*}
    T^{\otimes n} \ket{a}_L \; = \; \frac{1}{\sqrt{|C_2|}} \; \Bigg( \sum_{x \in C_2}  {(-1)}^{w_H(s_X \ast x)} \; e^{\iota \frac{\pi}{4} w_H(x \oplus y_a \oplus s_Z)}  \ket{y_a \oplus x \oplus s_Z} \Bigg) 
\end{equation*}
Expanding the phase exponent, we obtain:
\begin{align*}
    w_H( x \oplus y_a \oplus s_Z) \; &= \; w_H(x) + w_H(y_a \oplus s_Z) - 2w_H(x \ast \bracket{y_a \oplus s_Z})  \\
    &= \; w_H(y_a \oplus s_Z) \!\!\pmod{8}
\end{align*}
where the last equality holds by~\eqref{eqn_25}. Thus, the transversal $T$ preserves $\ket{a}_L \in \mathbb{Q}_S$:
\begin{equation*}
    T^{\otimes n} \ket{a}_L \; = \; e^{\iota \frac{\pi}{4} w_H(y_a \oplus s_Z)} \ket{a}_L \in \mathcal{Q}_S
\end{equation*}
Since $\left\{ \ket{a}_L : a \in \mathbb{F}_2^k \right\}$ forms a basis of the code space $\mathcal{Q}_S$, the transversal $T$ is a logical operator.

\textit{Necessity}: Assume that $T^{\otimes n}$ is a logical operator. For any $y \in C_1$, we evaluate $T^{\otimes n} E(y,0) \ket{0^k}_L$ in two ways. First, since both $T^{\otimes n}$ and $E(y,0)$ are logical operators, they commute with the projector, $P_{\cS}$. From~\eqref{eqn_20}, it follows that:
\begin{align}
    T^{\otimes n} \; E(y,0) \ket{0^k}_L \; &= \; \sqrt{|C_2|} \; T^{\otimes n} \; E(y, 0) \; P_{\cS}  \ket{s_Z} \nonumber \\
    &= \; \sqrt{|C_2|} \; P_{\cS}  \; T^{\otimes n} \; E(y,0)  \ket{s_Z} \nonumber \\
    &= \; e^{\iota \frac{\pi}{4} w_H(y \oplus s_Z)} \; E(y,0) \ket{0^k}_L \nonumber \\
    &= \; \frac{1}{\sqrt{|C_2|}} \; e^{\iota \frac{\pi}{4} w_H(y \oplus s_Z)} \Bigg(  \sum_{x \in C_2} \; {(-1)}^{w_H(s_X \ast x)} \ket{y \oplus x \oplus s_Z}  \Bigg) \label{eqn_33b}
\end{align}
Second, we compute the action directly by applying $T^{\otimes n}$:
\begin{align}
    T^{\otimes n} \; E(y,0) \ket{0^k}_L \; &= \; \frac{1}{\sqrt{|C_2|}} \; \Bigg(  \sum_{x \in C_2} {(-1)}^{w_H(s_X \ast x)} \; e^{\iota \frac{\pi}{4} w_H(y \oplus x \oplus s_Z)}  \ket{y \oplus x \oplus s_Z} \Bigg) \nonumber \\
    &= \; \frac{1}{\sqrt{|C_2|}} \; e^{\iota \frac{\pi}{4} w_H(y \oplus s_Z)} \; \Bigg(  \sum_{x \in C_2}  {(-1)}^{w_H(s_X \ast x)} \; e^{\iota \frac{\pi}{4} \bracket{w_H(x) - 2 w_H(x \ast \bracket{y \oplus s_Z})}}  \ket{y \oplus x \oplus s_Z}  \Bigg) \label{eqn_35b}
\end{align}
For any $y \in C_1$, $\left\{ \ket{y \oplus x \oplus s_Z} : x \in C_2  \right\}$ is a linearly independent set in ${\bracket{\mathbb{C}^2}}^{\otimes n}$. Therefore, by equating~\eqref{eqn_33b} and~\eqref{eqn_35b}, we get that: for any $x \in C_2$,
\begin{equation*}
    w_H(x) - 2 w_H(x \ast (y \oplus s_Z)) = 0 \!\!\pmod{8}
\end{equation*}
\end{proof}

The following corollary establishes a CSS-T characterization equivalent to Theorem~\ref{theorem_3}. This characterization will simplify the proof of Lemma~\ref{lemma_21}.
\begin{corollary}\label{corollary_4}
    A CSS code $\css{(C_1, C_2, s_Z)}$ is a CSS-T code iff for all $x \in C_2$, $y \in C_1 / C_2$,
    \begin{equation}
        w_H(x) - 2 w_H \big(x \ast \bracket{ y \oplus s_Z} \big) = 0 \!\!\pmod{8}\label{eqn_26b}
    \end{equation}
\end{corollary}
\begin{proof} \text{ }\\
The necessity follows immediately from~\eqref{eqn_25} of Theorem~\ref{theorem_3}. To prove sufficiency, we must show that~\eqref{eqn_26b} implies~\eqref{eqn_25}.

Note that any $y' \in C_1$ can be decomposed as $y' = x' \oplus y$, where $x' \in C_2$ and $y \in C_1/C_2$. For any $x \in C_2$, consider~\eqref{eqn_25}:
\begin{align}
    w_H(x) - 2 w_H\big(x \ast (y' \oplus s_Z) \big) &= w_H(x) - 2 w_H\big(x \ast (x' \oplus y \oplus s_Z)\big) \nonumber \\
    &= w_H(x) - 2w_H(x \ast x') -2 w_H(x \ast y) -2 w_H(x \ast s_Z) 
    + 4 w_H(x \ast x' \ast y) \nonumber \\
    &\hspace{1cm}+ 4 w_H(x \ast x' \ast s_Z) + 4 w_H(x \ast y \ast s_Z) \nonumber \\
    &= w_H(x) - 2 w_H(x \ast (y \oplus s_Z)) - 2w_H(x \ast x') + 4 w_H(x \ast x' \ast y) \nonumber \\
    &\hspace{1cm}+ 4 w_H(x \ast x' \ast s_Z) \nonumber \\
    &= - 2w_H(x \ast x') + 4 w_H(x \ast x' \ast y) + 4 w_H(x \ast x' \ast s_Z) \pmod{8} \ \  \label{eqn_42c}
\end{align}
where the last equality holds by~\eqref{eqn_26b}. Next, we apply the hypothesis~\eqref{eqn_26b} to $x \oplus x' \in C_2$ and $y \in C_1/C_2$:
\begin{align}
    0 \!\!\!\pmod{8}&= w_H(x \oplus x') - 2 w_H\big((x \oplus x') \ast (y \oplus s_Z) \big) \nonumber  \\
    &= w_H(x \oplus x') - 2 w_H \big( \bracket{x \ast (y \oplus s_Z)} \oplus \bracket{x' \ast (y \oplus s_Z)} \big) \nonumber \\
    &= w_H(x) + w_H(x') - 2w_H(x \ast x') - 2 w_H \big(x \ast (y \oplus s_Z)\big) - 2 w_H\big(x' \ast (y \oplus s_Z)\big)\nonumber \\
    &\hspace{1cm} + 4w_H\big(x \ast x' \ast (y \oplus s_Z)\big) \nonumber \\
    &= - 2w_H(x \ast x') + 4 w_H\big(x \ast x' \ast (y \oplus s_Z)\big) \pmod{8} \label{eqn_27a} \\
    &= -2 w_H(x \ast x') + 4 w_H(x \ast x' \ast y) + 4 w_H(x \ast x' \ast s_Z) \!\!\pmod{8}\label{eqn_47a}
\end{align}
where~\eqref{eqn_27a} follows from the hypothesis~\eqref{eqn_26b}. Substituting~\eqref{eqn_47a} into~\eqref{eqn_42c} completes the proof.
\end{proof}

\begin{proof}[Proof of Lemma~\ref{lemma_21}]\text{ }\\

    \textit{Sufficiency}: It suffices to show that~\eqref{eqn_28f}-\eqref{eqn_32f} implies~\eqref{eqn_26b}. For any $x = \oplus_{i=1}^{k_2} a_i x_i \in C_2$, $y = \oplus_{q=1}^k b_j y_j \in C_1 / C_2$, consider:
    \begin{align}
        w_H&(x) - 2 w_H \big( x \ast (y \oplus s_Z) \big) = w_H\Big( \oplus_{i=1}^{k_2} a_i x_i \Big) - 2w_H \Big( \Big(\oplus_{i=1}^{k_2} a_i x_i \Big) \ast \Big( \Big( \oplus_{q=1}^k b_q y_q \Big) \oplus s_Z  \Big) \Big)  \nonumber  \\
        &= w_H \Big( \oplus_{i=1}^{k_2} a_i x_i \Big) - 2w_H \Big( \Big( \oplus_{i=1}^{k_2} a_i x_i \ast s_Z \Big) \oplus \Big( \oplus_{i=1}^{k_2} \oplus_{q=1}^k a_i b_q x_i \ast y_q  \Big) \Big) \nonumber \\
        &= \sum_{i=1}^{k_2} a_i w_H(x_i) - 2 \sum_{i=1}^{k_2} \sum_{j=i+1}^{k_2} a_i a_j w_H(x_i \ast x_j) + 4 \sum_{i=1}^{k_2} \sum_{j=i+1}^{k_2} \sum_{p=j+1}^{k_2} a_i a_j a_p w_H(x_i \ast x_j \ast x_p) - 2 \sum_{i=1}^{k_2} a_i w_H(x_i \ast s_Z)  \nonumber \\
        &\hspace{0.8cm}- 2 \sum_{i=1}^{k_2} \sum_{q=1}^k a_i b _q w_H(x_i \ast y_q) + 4 \sum_{i=1}^{k_2} \sum_{j=i+1}^{k_2} a_i a_j w_H(x_i \ast x_j \ast s_Z) + 4 \sum_{i=1}^{k_2} \sum_{j=1}^{k_2} \sum_{q=1}^k a_i a_j b_q w_H(x_i \ast s_Z \ast x_j \ast y_q) \nonumber \\
        &\hspace{0.8cm}+ 4 \sum_{i=1}^{k_2} \sum_{q=1}^{k}  \sum_{r=q+1}^k a_i b_q b_r w_H(x_i \ast y_q \ast y_r) + 4 \sum_{i=1}^{k_2} \sum_{q=1}^{k} \sum_{j=i+1}^{k_2} \sum_{r=1}^{k} a_i a_j b_q b_r w_H(x_i \ast x_j \ast y_q \ast y_r) \!\!\!\pmod{8}\nonumber \\
        &= - 2 \sum_{i=1}^{k_2} \sum_{q=1}^k a_i b _q w_H(x_i \ast y_q)  + 4 \sum_{i=1}^{k_2} \sum_{j=1}^{k_2} \sum_{q=1}^k a_i a_j b_q w_H(x_i \ast s_Z \ast x_j \ast y_q) \nonumber \\ 
        &\hspace{0.8cm}+  4 \sum_{i=1}^{k_2} \sum_{q=1}^k \sum_{j=i+1}^{k_2} \sum_{r=1}^{k} a_i a_j b_q b_r w_H(x_i \ast x_j \ast y_q \ast y_r)  \!\!\!\pmod{8} \label{eqn_29b}   \\
        &= - 2 \sum_{i=1}^{k_2} \sum_{q=1}^k a_i b _q w_H(x_i \ast y_q)  + 4 \sum_{q=1}^{k} \Bigg( \sum_{i=1}^{k_2} a_i b_q w_H(x_i \ast s_Z \ast y_q) + 2 \sum_{i=1}^{k_2} \sum_{j=i+1}^{k_2}  a_i a_j b_q w_H(x_i \ast s_Z \ast x_j \ast y_q) \Bigg) \nonumber \\
        &\hspace{0.8cm} +  4 \sum_{i=1}^{k_2} \sum_{j=i+1}^{k_2} \Bigg( \sum_{q=1}^k  a_i a_j b_q  w_H(x_i \ast x_j \ast y_q) +  2\sum_{q=1}^k  \sum_{r=q+1}^{k} a_i a_j b_q b_r w_H(x_i \ast x_j \ast y_q \ast y_r) \Bigg) \!\!\!\pmod{8}\nonumber \\
        &=- 2 \sum_{i=1}^{k_2} \sum_{q=1}^k a_i b _q w_H(x_i \ast y_q)  + 4 \sum_{i=1}^{k_2} \sum_{q=1}^k a_i b_q w_H(x_i \ast s_Z \ast y_q) + 4 \sum_{i=1}^{k_2} \sum_{j=i+1}^{k_2} \sum_{q=1}^{k}  a_i a_j b_q  w_H(x_i \ast x_j \ast y_q) \!\!\!\pmod{8} \nonumber \\
        &= 0 \!\!\!\pmod{8} \label{eqn_30b}
    \end{align}
    where Eq.~\eqref{eqn_29b} follows from~\eqref{eqn_28f},~\eqref{eqn_29f},~\eqref{eqn_31f}, and~\eqref{eqn_32f}, while Eq.\eqref{eqn_30b} is obtained from~\eqref{eqn_30f} and~\eqref{eqn_33f}.
    
\textit{Necessity}: If ${\css(C_1, C_2, s_Z)}$ is a CSS-T code then  Corollary~\ref{corollary_4} implies that for all $x \in C_2$, $y \in C_1 / C_2$: 
\begin{equation}\label{eqn_60f}
    w_H(x) - 2w_H \big( x \ast (y \oplus s_Z) \big) = 0 \!\!\pmod{8}
\end{equation}
For $i \in [k_2]$, setting $x=x_i \in C_2$ and $y=0 \in C_1 / C_2$ in~\eqref{eqn_60f}, we obtain~\eqref{eqn_28f}:
\begin{eqnarray}
    w_H(x_i) - 2 w_H(x_i \ast s_Z) = 0 \!\!\!\pmod{8}\label{eqn_61f}
\end{eqnarray}
For $i \neq j \in [k_2]$, setting $x=x_i \oplus x_j \in C_2$, $y=0 \in C_1 / C_2$ in~\eqref{eqn_60f}, we obtain~\eqref{eqn_29f}:
\begin{align}
    0 \!\!\!\pmod{8}&= w_H(x_i \oplus x_j) - 2 w_H \big( (x_i \oplus x_j) \ast s_Z \big) \nonumber \\
    &= w_H(x_i) + w_H(x_j) - 2 w_H(x_i \ast x_j) -  2w_H(x_i \ast s_Z) -  2w_H(x_j \ast s_Z) +  4 w_H(x_i \ast x_j \ast s_Z) \nonumber \\
    &= -2w_H(x_i \ast x_j) + 4 w_H(x_i \ast x_j \ast s_Z) \!\!\!\pmod{8}  \label{eqn_64f}
\end{align} 
where Eq.~\eqref{eqn_64f} follows from~\eqref{eqn_61f}. 

For distinct $i,j,p \in [k_2]$, setting $x=x_i \oplus x_j \oplus x_p \in C_2$, $y=0 \in C_1 / C_2$ in~\eqref{eqn_60f}, we obtain~\eqref{eqn_31f}:
\begin{align}
    0 \!\!\!\pmod{8}&= w_H(x_i \oplus x_j \oplus x_p) - 2 w_H \big( (x_i \oplus x_j \oplus x_p) \ast s_Z \big) \nonumber \\
    &= w_H(x_i) + w_H(x_j) + w_H(x_p) - 2 w_H(x_i \ast x_j) -2w_H(x_i \ast x_p) - 2w_H(x_j \ast x_p) + 4 w_H(x_i \ast x_j \ast x_p) \nonumber \\
    &\hspace{0.9cm}-  2w_H(x_i \ast s_Z) -  2w_H(x_j \ast s_Z) -2w_H(x_p \ast s_Z)  +  4 w_H(x_i \ast x_j \ast s_Z) + 4 w_H(x_i \ast x_p \ast s_Z) \nonumber \\
    &\hspace{0.9cm}+ 4 w_H(x_j \ast x_p \ast s_Z)  \!\!\!\pmod{8} \nonumber  \\
    &= 4 w_H(x_i \ast x_j \ast x_p) \!\!\!\pmod{8} \label{eqn_34a}  
\end{align}
where Eq.~\eqref{eqn_34a} follows from~\eqref{eqn_61f} and~\eqref{eqn_64f}. 

For $i \in [k_2]$, $p \in [k]$, setting $x = x_i \in C_2$, $y = y_q \in C_1 / C_2$, we obtain~\eqref{eqn_30f}:
\begin{align}
   0 \!\!\!\pmod{8}&=  w_H(x_i) - 2 w_H \big( x_i \ast (y_q \oplus s_Z) \big) \nonumber \\
   &= w_H(x_i) - 2 w_H(x_i \ast y_q) - 2 w_H(x_i \ast s_Z) + 4 w_H(x_i \ast y_q \ast s_Z) \nonumber \\
   &= - 2 w_H(x_i \ast y_j) + 4 w_H(x_i \ast y_q \ast s_Z) \pmod{8} \label{eqn_70f}
\end{align}
where Eq.~\eqref{eqn_70f} follows from~\eqref{eqn_61f}.

For $ i\neq j \in [k_2]$, $p \in [k]$, setting $x = x_i \oplus x_j \in C_2$, $y = y_q \in C_1 / C_2$, we obtain~\eqref{eqn_33f}:
\begin{align}
    0 \!\!\!\pmod{8}&=  w_H(x_i \oplus x_j) - 2 w_H \big(  (x_i \oplus x_j) \ast (y_q \oplus s_Z) \big) \nonumber \\
    &= w_H(x_i) + w_H(x_j) -2w_H(x_i \ast x_j) - 2 w_H(x_i \ast y_q) - 2_H(x_i \ast s_Z) -2w_H(x_j \ast y_q) -2w_H(x_j \ast s_Z) \nonumber \\  
   &\hspace{0.9cm}+ 4w_H(x_i \ast y_q \ast s_Z) + 4 w_H(x_i \ast y_q \ast x_j) + 4 w_H(x_i \ast y_q \ast x_j \ast s_Z) + 4 w_H(x_i \ast s_Z \ast x_j \ast y_q) \nonumber \\
   &\hspace{0.9cm}+ 4 w_H(x_i \ast s_Z \ast x_j) + 4 w_H(x_j \ast y_q \ast s_Z) \pmod{8}\nonumber \\
   &= 4 w_H(x_i \ast y_q \ast x_j) \!\!\!\pmod{8} \label{eqn_36b}
\end{align}
where Eq.~\eqref{eqn_36b} follows from~\eqref{eqn_61f},~\eqref{eqn_64f} and~\eqref{eqn_70f}.

For $q \neq r \in [k]$, $ i \in [k_2]$, setting $x = x_i  \in C_2$, $y = y_q \oplus y_r \in C_1 / C_2$, we obtain~\eqref{eqn_32f}:
\begin{align}
    0 \!\!\!\pmod{8}&= w_H(x_i) - 2w_H \big(x_i \ast (y_q \oplus y_r \oplus s_Z) \big) \nonumber \\
    &= w_H(x_i) - 2w_H(x_i \ast y_q) -2 w_H(x_i \ast y_r) -2w_H(x_i \ast s_Z) + 4 w_H(x_i \ast y_q \ast y_r) \nonumber \\
    & \hspace{0.9cm}+ 4 w_H(x_i \ast y_q \ast s_Z)  + 4w_H(x_i \ast y_r \ast s_Z) \pmod{8} \nonumber \\
    &= 4 w_H(x_i \ast y_q \ast y_r) \!\!\!\pmod{8} \label{eqn_37}
\end{align}
where Eq.~\eqref{eqn_37} follows from~\eqref{eqn_61f} and~\eqref{eqn_70f}. This concludes the proof.
\end{proof}

\begin{proof}[Proof of Corollary~\ref{corollary_22}]\text{ }\\
    Let $\left\{  x_1, \ldots, x_{k_2}\right\}$ be a basis of $C_2$ and $\left\{ y_1, \ldots, y_k \right\}$ be a basis of $C_1 / C_2$. If ${\css{(C_1, C_2,s_Z)}}$ is a CSS-T code, Lemma~\ref{lemma_21} implies that for distinct $i,j,p \in [k_2]$, $q, r \in [k]$, the following holds:
    \begin{align}
         w_H(x_i) &= 0 \!\pmod{2}, & w_H(x_i \ast x_j) &=  0 \!\pmod{2}, &w_H(x_i \ast x_j \ast x_p) &=  0 \!\pmod{2}, \label{eqn_79f} \\
         w_H(x_i \ast y_q) &= 0 \!\pmod{2}, & w_H(x_i \ast x_j \ast y_q)  &= 0 \!\pmod{2}, & w_H(x_i \ast y_q \ast y_r) &= 0 \!\pmod{2}. \label{eqn_80f}
    \end{align}
    The Schur product $C_2 \ast C_1$ is spanned by the component-wise products of the basis vectors of $C_2$ and $C_1$:
    \begin{equation}\label{eqn_83f}
        C_2 \ast C_1 = \big\langle \left\{ x_i : i \in [k_2] \right\} \cup \left\{ x_i \ast x_j : i \neq j \in [k_2] \right\} \cup \left\{ x_i \ast y_q : i \in [k_2], q \in [k] \right\}  \big\rangle
    \end{equation}
    The conditions in~\eqref{eqn_79f} and~\eqref{eqn_80f} verify that each generator of $C_2 \ast C_1$ in~\eqref{eqn_83f} is orthogonal to all basis vectors of $C_1$. Consequently, $C_2 \ast C_1 \subseteq C_1^{\perp}$.
\end{proof}

\newpage 

\subsection{Proof of Theorem~\ref{theorem_19}, Theorem~\ref{theorem_20}, and Lemma~\ref{lemma_19}}\label{appendix_d}
The following lemma describes the logical action induced by the transversal $T$ on a CSS-T code.
\begin{lemma}\label{lemma_6}
    Let $\css{(C_1, C_2, s_Z)}$ be an ${[[n,k]]}_2$ CSS-T code. Let $\{ y_1, \ldots, y_k \}$ be a basis of ${C_1/C_2}$. For any $a = (a_1, \ldots, a_k) \in \mathbb{F}_2^k$, let $y_a = \oplus_{i=1}^k a_i y_i \in C_1 / C_2$. The logical action of the transversal $T$ on the encoded state $\ket{a}$ is given by:
    \begin{equation*}
       T^{\otimes n} \ket{a}_L \; = \; \underbrace{e^{\iota \frac{\pi}{4} w_H(s_Z)}}_{\text{global phase}} \cdot \underbrace{e^{\iota \frac{\pi}{4} \bracket{w_H(y_a) - 2 w_H(y_a \ast s_Z)}}}_{\text{phase on logical qubits}} \ket{a}_L
    \end{equation*}
\end{lemma}
\begin{proof}\text{ }\\
    Since $T^{\otimes n}$ and $E(y_a, 0)$ are logical operators, they commute with the code projector. From~\eqref{eqn_20} and~\eqref{eqn_24}, it follows that:
    \begin{align*}
        T^{\otimes n} \ket{a}_L &= \; \sqrt{|C_2|} \; T^{\otimes n} \; E(y_a, 0) \; P_{\cS}  \ket{s_Z} \\
        &= \; \sqrt{|C_2|} \; P_{\cS} \; T^{\otimes n} \; E(y_a, 0) \ket{s_Z} \\
        &= \; e^{\iota \frac{\pi}{4} w_H(y_a \oplus s_Z)} \; \sqrt{|C_2|} \; E(y_a, 0) \;  P_{\cS}  \ket{s_Z} \\
        &= \; e^{\iota \frac{\pi}{4} w_H(s_Z)} \; e^{\iota \frac{\pi}{4} \bracket{w_H(y_a) -2w_H(y_a \ast s_Z)}} \ket{a}_L 
    \end{align*}
\end{proof}

\begin{proof}[Proof of Theorem~\ref{theorem_19}]\text{ }\\
    First, observe that if an operator acts as the logical identity in one basis, it does so in any basis. Let $\{y_1, \ldots, y_k\}$ be a basis of the coset space $C_1 / C_2$.
    
    \textit{Sufficiency}: Given $\css{(C_1, C_2,s_Z)}$ is a CSS-T code. By Lemma~\ref{lemma_6}, for any $a  = (a_1, \ldots, a_k)\in \mathbb{F}_2^k$ and $y_a = \oplus_{i=1}^{k}a_i y_i \in C_1 /C_2$, we have:
    \begin{align*}
        T^{\otimes n}  \ket{a}_L \; &= \; e^{\iota \frac{\pi}{4} w_H(s_Z)}  \; e^{\iota \frac{\pi}{4} (w_H(y_a) - 2w_H(y_a \ast s_Z))} \ket{a}_L \\
        &= \; e^{\iota \frac{\pi}{4} w_H(s_Z)} \ket{a}_L 
    \end{align*}
    where the above equality holds by~\eqref{eqn_55}. Thus, $T^{\otimes n}$ acts as the logical identity with global phase $e^{\iota \frac{\pi}{4} w_H(s_Z)}$. 
    
    \textit{Necessity}: Suppose $T^{\otimes n}$ realizes the logical identity. Then there exists a global phase $\gamma_T$ such that $T^{\otimes n} \ket{a}_L = \gamma_T \ket{a}_L$ for all $a \in \mathbb{F}_2^k$.
    First, consider the zero logical state $\ket{0^k}$ (corresponding to $y_a = 0$). Lemma~\ref{lemma_6} implies:
    \begin{equation*}
        T^{\otimes n}  \ket{0^k}_L \; = \; e^{\iota \frac{\pi}{4} w_H(s_Z)}  \ket{0^k}_L \implies \gamma_T \; = \; e^{\iota \frac{\pi}{4} w_H(s_Z)}
    \end{equation*}
    Now consider any $y \in C_1/C_2$. Let $\ket{a}$ be the logical state associated with $y$ (i.e., $y_a = y$). Equating the action from Lemma~\ref{lemma_6} with $\gamma_T \ket{a}_L$:
    \begin{equation*}
        \gamma_T \; e^{\iota \frac{\pi}{4} (w_H(y) - 2w_H(y \ast s_Z))}  \ket{a}_L \; = \;  \gamma_T \ket{a}_L
    \end{equation*}
    Since $\ket{a}_L \neq 0$, we obtain:
    \begin{equation*}
        w_H(y) - 2w_H(y \ast s_Z) = 0 \!\pmod{8}
    \end{equation*}
\end{proof}

\begin{proof}[Proof of Theorem~\ref{theorem_20}] \text{ }\\
\textit{Sufficiency}: Given $\css{(C_1, C_2,s_Z)}$ is a CSS-T code. By Lemma~\ref{lemma_6}, for any $a  = (a_1, \ldots, a_k)\in \mathbb{F}_2^k$ and $y_a = \oplus_{i=1}^{k}a_i y_i \in C_1 /C_2$, we have:
\begin{align*}
    T^{\otimes n} \ket{a}_L \; &= \; e^{\iota \frac{\pi}{4} w_H(s_Z)}  \; e^{\iota \frac{\pi}{4} \bracket{w_H(y_a) - 2w_H(y_a \ast s_Z)}} \ket{a}_L \\
    &= \; e^{\iota \frac{\pi}{4} w_H(s_Z)} \; e^{\iota \frac{\pi}{4} w_H(a)} \ket{a}_L 
\end{align*}
where the above equality holds by~\eqref{eqn_17}. Thus, $T^{\otimes n}$ acts as the logical transversal $T$ with global phase $e^{\iota \frac{\pi}{4} w_H(s_Z)}$. 

 \textit{Necessity}: Suppose $T^{\otimes n}$ realizes the logical transversal $T$. Then there exists a global phase $\gamma_T$ such that $T^{\otimes n}\ket{a}_L  = \gamma_T \; e^{\iota \frac{\pi}{4} w_H(a)}  \ket{a}_L$ for all $a \in \mathbb{F}_2^k$.
    First, consider the zero logical state $\ket{0^k}$ (corresponding to $y_a = 0$). Lemma~\ref{lemma_6} implies:
    \begin{equation*}
        T^{\otimes n} \ket{0^k}_L \; = \; e^{\iota \frac{\pi}{4} w_H(s_Z)} \ket{0^k}_L \implies \gamma_T\;  = \; e^{\iota \frac{\pi}{4} w_H(s_Z)}
    \end{equation*}
    Now consider any $a\in \mathbb{F}_2^k$ and the corresponding coset representative $y_a \in C_1/C_2$. Equating the action from Lemma~\ref{lemma_6} with $\gamma_T \; e^{\iota \frac{\pi}{4} w_H(a)} \ket{a}_L$:
    \begin{equation*}
        \gamma_T \; e^{\iota \frac{\pi}{4} \bracket{w_H(y_a) - 2w_H(y_a \ast s_Z)}} \ket{a}_L \; = \;  \gamma_T \; e^{\iota \frac{\pi}{4} w_H(a)}  \ket{a}_L
    \end{equation*}
    Since $\ket{a}_L \neq 0$, we obtain:
    \begin{equation*}
        w_H(y_a) - 2w_H(y_a \ast s_Z) \; = \; w_H(a) \!\pmod{8}
    \end{equation*}
\end{proof}

\begin{proof}[Proof of Lemma~\ref{lemma_19}]\text{ }\\
    It suffices to show that $\css{({C}_1^{(1)}, {C}_2^{(1)},s_Z =(1, 0))}$, satisfies the CSS-T characterization in~\eqref{eqn_25}. For any $\widetilde{x}=(x,x) \in {C}_2^{(1)}, \ \ \widetilde{y}=(y,y) \in {C}_1^{(1)}$, we have:
    \begin{align*}
        w_H(\widetilde{x})  -  2 w_H \big(\widetilde{x} \ast (\widetilde{y} \oplus s_Z) \big) \; = \; 2 w_H(x)  -  2 w_H \big(x \ast (y \oplus {1}), x \ast y \big) \; = \; 0
    \end{align*}
    To show that the transversal $T$ realizes the logical identity, it suffices to verify~\eqref{eqn_55}. Indeed, for any $\widetilde{z} = (z,z) \in {C}_1^{(1)} / {C}_2^{(1)}$:
    \begin{equation*}
        w_H(\widetilde{z})  -  2 w_H(\widetilde{z} \ast s_Z) \; = \; 2w_H(z)  -  2w_H(z) \; = \; 0
    \end{equation*}
\end{proof}
\end{document}